\documentclass[twocolumn,english,superscriptaddress,prl]{revtex4-1}
\usepackage{graphicx}% Include figure files
\usepackage{bm}% bold math
\usepackage{amsmath}
\usepackage{amsfonts,dsfont}
\usepackage{bbm} %bold numbers
\usepackage[FIGTOPCAP]{subfigure}
\usepackage{xr} %allow cross-referencing across files
\usepackage{times}
\usepackage{color}
\usepackage[colorlinks=true,citecolor=blue]{hyperref}
\usepackage{physics}
\usepackage{algpseudocode}
\usepackage{algorithm}
\usepackage{mathtools}
\usepackage{amsthm}
\usepackage{qcircuit}
\usepackage{complexity}

\newcommand{\ra}{\ensuremath{\rightarrow}}

\newcommand{\al}{\ensuremath{\alpha}}

\renewcommand{\O}[1]{\mathcal{O}(#1)}
\renewcommand{\W}[1]{\Omega(#1)}
\newcommand{\Th}[1]{\Theta(#1)}

\newcommand{\dash}{\textemdash}

\newcommand{\sect}[1]{\emph{{#1}}.\dash}

\newtheorem*{theorem*}{Theorem}

\newtheorem*{lemma*}{Lemma}

\usepackage[capitalise,compress]{cleveref}
\crefname{section}{Sec.}{Secs.}
\crefname{section}{Section}{Sections}
\crefrangelabelformat{equation}{\textup{(#3#1#4)}--\textup{(#5#2#6)}}
\newcommand{\tghz}{t_{\text{GHZ}}}

\makeatletter
\newcommand\footnoteref[1]{\protected@xdef\@thefnmark{\ref{#1}}\@footnotemark}
\makeatother

\algblock{ParFor}{EndParFor}
\algnewcommand\algorithmicparfor{\textbf{parfor}}
\algnewcommand\algorithmicpardo{\textbf{do}}
\algnewcommand\algorithmicendparfor{\textbf{end\ parfor}}
\algrenewtext{ParFor}[1]{\algorithmicparfor\ #1\ \algorithmicpardo}
\algrenewtext{EndParFor}{\algorithmicendparfor}

\begin{document}

\title{Implementing a Fast Unbounded Quantum Fanout Gate Using Power-Law Interactions}
\author{Andrew~Y.~Guo}
\author{Abhinav~Deshpande}
\author{Su-Kuan~Chu}
\author{Zachary~Eldredge}
\author{Przemyslaw~Bienias}
\author{Dhruv~Devulapalli}
\affiliation{Joint Center for Quantum Information and Computer Science, NIST/University of Maryland, College Park, MD 20742, USA}
\affiliation{Joint Quantum Institute, NIST/University of Maryland, College Park, MD 20742, USA}

\author{Yuan~Su}
\author{Andrew~M.~Childs}
\affiliation{Joint Center for Quantum Information and Computer Science, NIST/University of Maryland, College Park, MD 20742, USA}
\affiliation{Department of Computer Science, University of Maryland, College Park, MD 20742, USA}
\affiliation{Institute for Advanced Computer Studies, University of Maryland, College Park, MD 20742, USA}

\author{Alexey~V.~Gorshkov}
\affiliation{Joint Center for Quantum Information and Computer Science,
NIST/University of Maryland, College Park, MD 20742, USA}
\affiliation{Joint Quantum Institute, NIST/University of Maryland, College Park, MD 20742, USA}

\date{\today}

\begin{abstract}

The standard circuit model for quantum computation presumes the ability to directly perform gates between arbitrary pairs of qubits, which is unlikely to be practical for large-scale experiments.
Power-law interactions with strength decaying as $1/r^\alpha$ in the distance $r$ provide an experimentally realizable resource for information processing, whilst still retaining long-range connectivity.
We leverage the power of these interactions to  implement a fast quantum fanout gate with an arbitrary number of targets.
Our implementation allows the quantum Fourier transform (QFT) and Shor's algorithm to be performed on a $D$-dimensional lattice in time logarithmic in the number of qubits for interactions with $\al$\,$\le$\,$D$.
As a corollary, we show that power-law systems with $\alpha$\,$\le$\,$D$ are difficult to simulate classically even for short times, under a standard assumption that factoring is classically intractable.
Complementarily, we develop a new technique to give a general lower bound, linear in the size of the system, on the time required to implement the QFT and the fanout gate in systems that are constrained by a linear light cone.
This allows us to prove an asymptotically tighter lower bound for long-range systems than is possible with previously available techniques.
\end{abstract}

\maketitle

In the standard circuit model for quantum computation, the size of a quantum circuit is measured in terms of the number of gates it contains.
In typical quantum systems, coherence times are a limitation, so low-depth (``shallow'') quantum circuits prioritized for the regime of noisy intermediate-scale quantum computers are more desirable \cite{Preskill2018}.
Various proposed models of quantum computation are equivalent up to polynomial overhead, making the definition of the complexity class $\BQP$ insensitive to the model of computation \cite{Bernstein1993,Raussendorf2001,Raussendorf2003,Hoyer2005}.

However, these models can differ in the precise complexity of operations.
As a drastic example, suppose we are given access to a fast unbounded fanout gate
represented by the map
$\ket{x} \ket{y_1} \ket{y_2} \dots$\,$\mapsto$\,$\ket{x}\ket{y_1 \oplus x} \ket{y_2 \oplus x} \cdots$
where the $\oplus$ operator denotes bitwise XOR (bit $y_i$ is flipped if $x$\,$=$\,$1$ and not flipped otherwise). This operation is a reversible analog of a gate that copies $x$ to registers $y_1, y_2, \dots$.
By ``unbounded,'' we mean that there is no limit on the number of bits that can be targeted by this operation.

The unbounded fanout gate makes it possible for constant-depth quantum circuits to perform a number of fundamental quantum arithmetic operations \cite{Hoyer2005}.
Furthermore, unbounded fanout can also reduce the quantum Fourier transform (QFT)---a subroutine of a large class of quantum algorithms, including most famously Shor’s algorithm for integer factorization \cite{Shor1997}---to constant depth as well.
In fact, it enables implementing the entirety of Shor's algorithm by constant-depth quantum circuits with access to a polynomial amount of classical pre- and post-processing \cite{note_precomputation}.

While the unbounded fanout gate is clearly a powerful resource for quantum computation, its efficient implementation in physically realizable architectures has not been studied in great depth.
In the standard circuit model of digital quantum computation---where one may apply single-qubit and two-qubit gates from a standard gate set on arbitrary non-overlapping subsets of the qubits---a fanout gate on $n$ qubits can be implemented optimally in $\Th{\log{n}}$-depth \cite{bigO,Broadbent2009a}.
We also consider the Hamiltonian model, in which one may apply single-qubit and two-qubit Hamiltonian terms.
In particular, in the Hamiltonian model with all-to-all unit-strength interactions, one can implement the fanout gate in constant time \cite{Fenner2003,Fenner2004}.
However, the assumption of being able to directly apply an interaction between two arbitrarily distant qubits does not hold in practice for large quantum computing architectures \cite{Monroe2014,Linke2017,Bapat2018,Childs2019c}.
Mapping these circuits to restricted architectures inevitably leads to overheads and potentially even different asymptotic scaling.
In $D$-dimensional nearest-neighbor architectures, for example, the unbounded fanout gate can only be implemented unitarily in depth $\Th{n^{1/D}}$ \cite{Rosenbaum2013}.
While there exist protocols that can implement the fanout gate in constant depth on these architectures \cite{Pham2013}, these proposals require intermediate measurements along with classical control---a resource that may be inaccessible in certain near-term experimental systems \cite{Arute2019}.
The overheads resulting from such physical restrictions could therefore limit the potential asymptotic speed-up from a fast quantum fanout.

Systems with power-law interactions present an opportunity for realizing these speed-ups.
Specifically, for a lattice of qubits in $D$ dimensions, suppose the interaction strengths between pairs of qubits separated by a distance $r$ are weighted by a power-law decaying function $1/r^\alpha$.

These long-range interactions are native to many experimental quantum systems and have attracted interest as potential resources for faster quantum information processing. Examples of long-range interactions include dipole-dipole and van der Waals interactions between Rydberg atoms~\cite{Saffman2010,Weimer2012}, and dipole-dipole interactions between polar molecules~\cite{Yan2013} and between defect centers in diamond~\cite{Yao2012,Weimer2012}.
Previous works have explored the acceleration of quantum information processing using strong and tunable power-law interactions between Rydberg states~\cite{Isenhower2011,Molmer2011,Petrosyan2017,Gulliksen2015,Muller2009,Young2020,Levine2019}, which can implement $k$-local gates that control or target simultaneously $k \gg 10$ qubits.
Those gates still have a finite spatial range and can therefore only give a constant-factor speed-up over nearest-neighbor architectures.
Recently, Refs.~\cite{Eldredge2017,Tran2020,Kuwahara2019a} gave protocols that take advantage of power-law interactions to quickly transfer a quantum state across a lattice.
As we will show, it is also possible to leverage the power of these interactions to implement quantum gates asymptotically faster than is possible with finite-range interactions \cite{note_ions}.

In this Letter, we describe a method of implementing the unbounded fanout gate using engineered Hamiltonians with power-law interactions.
As an application of this protocol, we show that simulating strongly long-range systems with $\al$\,$\le$\,$D$ for logarithmic time or longer is classically intractable, if factoring is classically hard.
We also develop a new technique that allows us to prove the tightest known lower bounds for the time required to implement the QFT and unbounded fanout in general lattice architectures.

\sect{Protocol for fast fanout using long-range interactions}We
use a modified version of the state transfer protocol from Ref.~\cite{Eldredge2017} to perform a fanout gate on $n$ logical qubits in $\O{\log{n}}$ time using long-range interactions with $\al$\,$=$\,$D$ and in $\O{1}$ time for $\al$\,$<$\,$D$.

As an intermediate step, the state-transfer protocol ``broadcasts'' a single-qubit state into the corresponding Greenberger–Horne–Zeilinger (GHZ)-like state:
\begin{equation}
    \label{eq:longrangebroadcast}
	(\psi_0 \ket{0} + \psi_1 \ket{1})\otimes \ket{0 0 \dots 0} \mapsto \psi_0 \ket{0 0  \dots 0} + \psi_1 \ket{1 1 \dots 1},
\end{equation}
where $\psi_0, \psi_1$\,$\in$\,$\mathbb{C}$ and $|\psi_0|^2$\,$+$\,$|\psi_1|^2$\,$=$\,$1$.
This long-range broadcast is achieved by performing a sequence of cascaded controlled-\textsc{Not} (\textsc{CNot}) gates---similar to the standard gate-based implementation of the unbounded fanout gate.
The \textsc{CNot} gate from qubit $i$ to qubit $j$ can be implemented by a Hamiltonian $H_{ij}$\,$=$\,$h_{ij} \ketbra{1}_i$\,$\otimes$\,$X_j$ acting for time $t$\,$=$\,$\pi/(2h_{ij})$, up to a local unitary \cite{Eldredge2017}.
Applying a Hamiltonian $H(t)$\,$=$\,$\sum_{ij} H_{ij}(t)$, which variously turns on/off interactions between pairs of qubits at different times, allows one to implement the broadcast in \cref{eq:longrangebroadcast}.

By using Hamiltonians with long-range interactions $h_{ij}$ satisfying $\norm{h_{ij}}$\,$\le$\,$1/r_{ij}^\al$, it is possible to implement the broadcast operation in sublinear time.
For a system of $n$ qubits, this broadcast time depends on the power-law exponent $\al$ and the dimension of the system $D$ as follows \cite{Eldredge2017}:
\begin{equation}
\label{eq:statetransfertimes}
	\tghz \propto
  \begin{cases}
    n^0 & \alpha < D
    \\ \log n & \alpha = D
    \\ n^{ ( \alpha - D)/D } & D < \alpha \le D + 1
    \\ n^{1/D} & \alpha > D + 1.
  \end{cases}
\end{equation}
We term the broadcast time $t_{\text{GHZ}}$, since it corresponds to the GHZ-state-construction time when $\psi_0$\,$=$\,$\psi_1$\,$=1/\sqrt{2}$.
This long-range broadcast is not the same as fanout because it requires that all intermediary qubits (besides the first qubit) be initialized in the $\ket{0}$ state.
However, as we now show, it is possible to adapt the broadcast protocol to implement the fanout gate in time $\tghz$ using $n$ ancillary qubits.

\begin{figure}[t]
  \begin{centering}
    \includegraphics[scale=0.45]{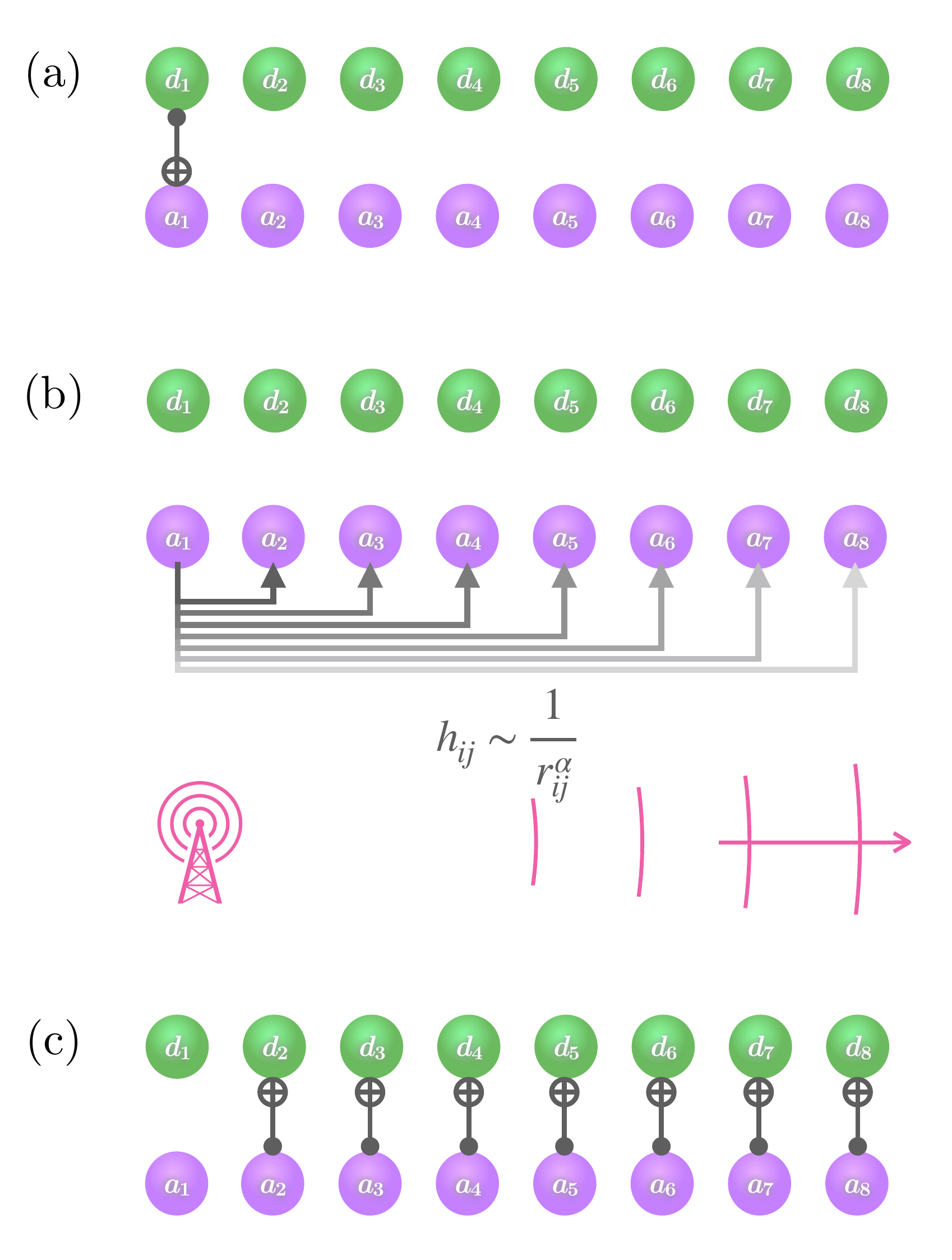}
  \par\end{centering}
  \caption{A protocol for a fast unbounded quantum fanout gate using long-range interactions, depicted here for a 1D lattice.
  The layout consists of a chain of data qubits, along with their adjacent ancillary qubits that are initialized to $\ket{0}$.
  (a) The first step is a local controlled-\textsc{Not} (\textsc{CNot}) gate from $\ket{d_1}$ to $\ket{a_1}$.
  (b) The application of the long-range ``broadcast'' from $\ket{a_1}$ to the rest of the ancillary qubits $\ket{a_i}$ creates a GHZ-like state in Eq.\,(\ref{eq:longrangebroadcast}) for the ancillary qubits together with the first data qubit.
  (c) We apply \textsc{CNot} gates from ancillary qubit $\ket{a_i}$ to the data qubit $\ket{d_i}$, which can be done in parallel. After this step, we reverse process (b) and process (a) to return the ancillary qubits to $\ket{0}$ (not redrawn here). }
  \label{fig:broadcast}
\end{figure}

Consider a system of $n$ qubits arranged on a $D$-dimensional lattice.
Furthermore, assume there are $n$ ancillary qubits, each located adjacent to one of the $n$ original data qubits.
Let us describe the qubits as $\ket{d_1}, \ket{d_2}, \dots, \ket{d_n}$ and $\ket{a_1}, \ket{a_2}, \dots, \ket{a_n}$ for data and ancilla, respectively.
Suppose we want to perform fanout with $\ket{d_1}$ as control, and that all ancillae are guaranteed to be in state $\ket{0}$.
Then the following sequence of operations (depicted graphically in \cref{fig:broadcast}) implements the fanout operation:
\begin{algorithm}[H]
  \caption{Implementing fanout with long-range interactions}
  \label{alg:fastfanout}
  \begin{algorithmic}[1]
   \State Initialize ancillary qubits: $\ket{a_i}$\,$\gets$\,$\ket{0}$ for $i$\,$=$\,$1$ to $n$
  \State \textsc{CNot}($\ket{d_1}$\,$\ra$\,$\ket{a_1}$)
  \Statex $\bigtriangledown$ Apply broadcast operation as shown in \cref{eq:longrangebroadcast} to ancillae:
  \State \textsc{LongRangeBroadcast}($\ket{a_1}$\,$\ra$\,$\ket{a_2},\dots,\ket{a_n}$)
  \Statex $\bigtriangledown$ \textbf{parfor} indicates that \textbf{for}-loop can be implemented in \textbf{par}allel
  \ParFor {$i$\,$=$\,$2$ to $n$}
    \State \textsc{CNot}($\ket{a_i}$\,$\ra$\,$\ket{d_{i}}$)\Comment{\parbox[t]{.45\linewidth}{{Transfer} fanout to data qubits}}
  \EndParFor
  \Statex $\bigtriangledown$ Apply broadcast operation in reverse to uncompute ancillae
  \State \textsc{ReverseLongRangeBroadcast}($\ket{a_2},\dots,\ket{a_n}$\,$\ra$\,$\ket{a_1}$)
  \State \textsc{CNot}($\ket{d_1}$\,$\ra$\,$\ket{a_1}$)
  \end{algorithmic}
\end{algorithm}

In addition to accomplishing fanout, this protocol returns the ancillary qubits to the $\ket{0}$ state.
Modulo $\O{n}$ short-range operations that can be done in parallel in one time step, the protocol requires time $2 t_\mathrm{GHZ}$.
Hence, it can implement the fanout gate in time that is constant for $\alpha$\,$<$\,$D$, logarithmic for $\alpha$\,$=$\,$D$, and polynomial for $\alpha$\,$>$\,$D$.

We briefly comment on the constant-depth implementation of the QFT and Shor's algorithm using the unbounded fanout gate.
An $n$-qubit QFT circuit can be performed with $\O{n\log n}$ gates to $1/\text{poly}(n)$ precision \cite{Coppersmith94}.
Using unbounded fanout, the circuit can be reduced to constant depth with $\O{n\log n}$ ancillary qubits \cite{Hoyer2005}.
We note that including these ancillae in the lattice would not change the asymptotic scaling of our protocol for $\al$\,$\le$\,$D$, since $\tghz$ is either $\O1$ or $\O{\log n}$ in this regime.

\sect{Intractability of classical simulation of strongly long-range systems}As
a corollary, the protocol shows that strongly long-range interacting systems with $\al$\,$\le$\,$D$ evolving for time logarithmic in $n$ or longer are difficult to simulate classically.
By this we mean that to approximately sample from the time-evolved state to within constant total-variation-distance error $\varepsilon$, a classical computer requires time at least superpolynomial in the worst case \cite{note_constanterror}.
The argument operates by a complexity-theoretic reduction from integer factoring, a problem that is assumed to be difficult for classical computers with the ability to use random bits (\textsc{Factoring} $\notin$\,$\BPP$).
The time required to implement the fanout gate using \cref{alg:fastfanout} is $t_\mathrm{FO}$\,$=$\,$\O{1}$ for $\al$\,$<$\,$D$ and $\O{\log n}$ for $\al$\,$=$\,$D$.
It is possible to implement Shor's order-finding algorithm in time $\O{t_\mathrm{FO}}$ using a small amount of classical pre-processing (polynomial in $n$) \cite{Cleve2000,Hoyer2005}.
Using the ability to sample from the output of the order-finding algorithm to error $\varepsilon$\,$<$\,$0.4$\,$<$\,$4/\pi^2$, classically efficient post-processing can output a factor of an $n$-bit integer with probability $\W{1}$ \cite{Shor1997}.
Therefore, if it were possible to efficiently sample from the output distribution in strongly long-range systems for evolution-time $t$\,$=$\,$\mathcal{O}(\log n)$, then it would be possible to factor $n$-bit integers efficiently as well.
The best classical algorithm currently known for factoring an $n$-bit integer takes runtime $\exp[\mathcal{O}(\sqrt{n \log n})]$ \cite{Lenstra1992} and the problem is widely believed to be classically intractable.
This stands in contrast to systems with finite-range interactions in 1D, for which efficient classical simulation is possible up to any time satisfying $t$\,$\leq$\,$\O{\log n}$ \cite{Osborne2006}.
Under the complexity assumption mentioned above (\textsc{Factoring} $\notin$\,$\BPP$), we have shown that this result is not fully generalizable to strongly long-range interacting systems.

\sect{Lower bounds on the time required to implement QFT and fanout}As a way to benchmark our long-range protocol for fanout, we discuss circuit-depth lower bounds in general lattice systems.
Recall that the protocol in \cref{alg:fastfanout} can implement fanout---and by corollary, the QFT---in time $\tghz$, which scales as $\O{\log n}$ for long-range systems with $\alpha$\,$=$\,$D$ and as $\O1$ for $\al$\,$<$\,$D$.
In this section, we show that such fast asymptotic runtimes cannot be achieved in architectures with strict locality constraints.

In Ref.\,\cite{Maslov2007}, Maslov showed that a specific way of implementing the QFT requires $\W n$ depth on the 1D  nearest-neighbor architecture, though this did not rule out other QFT implementations with smaller depth.
Here, we devise a technique that yields a lower bound of $\W {n^{1/D}}$ for the time required to perform a QFT in the Hamiltonian model.
This result strengthens and generalizes Maslov's bound to higher dimensions and to the Hamiltonian model.
In addition, we show that the same lower bound applies even to circuits that perform the QFT approximately.

Our lower bound holds for any lattice system with finite velocities of information spreading, which include short-range interactions (i.e., finite-ranged or exponentially decaying) and power-law interactions with $\al$\,$>$\,$5/2$ in one dimension or $\alpha$\,$>$\,$2D$\,$+$\,$1$ for $D$\,$>$\,$1$ \cite{Lieb1972,Chen2019,Kuwahara2019a}.
Combined with our results above, this implies that systems with strongly long-range interactions can implement the QFT and fanout asymptotically faster than more weakly interacting systems.

The intuitive idea behind our proof is that the QFT unitary can spread out operators in a certain precise sense, a task that can be bounded by the ``Frobenius-norm light cone'' of Ref.\,\cite{Tran2020}.
The fact that this light cone imposes a finite speed limit on information propagation in short-range interacting systems implies that the minimum time $t_2(r)$ required for operator-spreading is proportional to the distance between qubits $r$.
This constrains the computation time for the QFT, denoted $t_\mathrm{QFT}$, by $\W {n^{1/D}}$, from which the evolution-time (and hence circuit-depth) lower bound follows.
The same Frobenius-norm bound also constrains the time required to implement the approximate QFT (AQFT).

We consider the $4^n$-dimensional vector space of $n$-qubit operators for which the set of Pauli operators $\{I,X,Y,Z\}^{\otimes n}$ forms a basis.
We quantify operator spreading outside a region of radius $r$ as follows.
Taking an operator $|O)$ initially supported on site 1, we measure the weight of its time-evolved version, $|O(t))$, on sites at distance $r$ (and beyond) using a \emph{projection} operator $\mathcal{Q}_{r}$, which projects onto strings of Pauli operators that act nontrivially on at least one site at distance $r$ or greater.
We measure the weight of this projected operator $|O_r)$\,$\coloneqq$\,$\mathcal{Q}_{r}|O(t))$ via the (squared) normalized Frobenius norm $\norm{O_r}_F^2$\,$\coloneqq$\,$\Tr(O_r^\dag O_r)/2^n$, which coincides with the Euclidean norm over the operator space, $(O_r|O_r)$ \cite{note_norm}.

We show that operators spread by the action of the QFT can have high weight on distant regions,
which implies that $t_\mathrm{QFT}$\,$\geq$\,$t_2(r)$.
\begin{lemma*} \label{lem_qft_weight}
Let $U_\mathrm{QFT}$ be the QFT operator on $n$ qubits arranged in $D$ dimensions such that the first and $n$th qubits are a distance $r$\,$=$\,$\Theta(n^{1/D})$ apart.
Then $U_\mathrm{QFT}^\dag Z_1 U_\mathrm{QFT}$\,$=:$\,$Z_1' $ is an operator with at least constant weight at distance $r$.
\end{lemma*}

\begin{proof}
We explicitly compute the weight of the operator $Z_1'$ on site $n$.
Define $\omega$\,$:=$\,$e^{2\pi i/2^n}$.
The QFT operation on $n$ qubits is defined as $\sum_{y,z=0}^{2^n-1} \ketbra{y}{z} {\omega^{yz}}/{\sqrt{2^n}}$, where we interpret the bit string $y_1, y_2, \ldots, y_n$ as a binary representation of a number $y$\,$\in$\,$\{0,1,\ldots, 2^n -1\}$ in the canonical ordering, i.e.\ $y$\,$=$\,$y_1 2^{n-1}$\,$+$\,$y_2 2^{n-2}$\,$+$\,$\cdots$\,$+$\,$y_n$.
The inverse of the QFT is obtained simply by taking $\omega$\,$\rightarrow$\,$\omega^{-1}$.
First, we compute
\begin{align}
Z_1' &= \sum_{x,y,z=0}^{2^n-1} \ketbra{z}{y} \frac{\omega^{-z y}(-1)^{y_1} \omega^{yx}}{2^n} \ketbra{y}{x}.
\end{align}
We divide the sum over $y$ into two cases, $y_1$\,$=$\,$0$ and $y_1$\,$=$\,$1$:
\begin{align}
Z_1' &= \frac{1}{2^n}\sum_{x,z=0}^{2^n-1} \ketbra{z}{x} \left(\sum_{y:\, y_1 = 0} \omega^{(x-z) y} - \sum_{y:\, y_1= 1} \omega^{(x-z) y}\right).
\end{align}
We can compute these sums separately, giving
\begin{align}
\label{eq:zoneprime}
Z_1' = \frac{1}{2^{n-1}}\sum_{x\neq z}\ketbra{z}{x} \frac{1 - (-1)^{(x-z)}}{1-\omega^{x-z}}.
\end{align}
The nonzero terms in the sum on the right-hand side of \cref{eq:zoneprime} occur when $x-z$ is odd, i.e., when $x_n$\,$-$\,$z_n$\,$=$\,$1 \bmod 2$.
Therefore, the only terms that remain are off-diagonal on qubit $n$ or, equivalently, contain only the $X_n$ or $Y_n$ Pauli operators.
This implies that $Z_1'$ has all its weight on operators acting nontrivially at distance $r$---formally, that $\mathcal{Q}_r |Z_1')$\,$=$\,$|Z_1')$.
\end{proof}

As a result of the Lemma, $t_\mathrm{QFT}$ follows the light cone defined by the normalized Frobenius norm, which is at least as stringent as the Lieb-Robinson light cone.
This leads directly to the following theorem:
\begin{theorem*}  \label{thm:QFTlowerbounds}
  For systems with finite-range or exponentially-decaying interactions in $D$ dimensions, the time required to implement the QFT unitary is lower bounded by $t_\mathrm{QFT}$\,$= $\,$\Omega(r)$, where $r$\,$=$\,$\Theta(n^{1/D})$ is the distance between the first and $n$th qubits.

  For systems with long-range interactions, the Lieb-Robinson light cone gives the following bounds \cite{Kuwahara2019a,Tran2019,Hastings2006}:
  \begin{align}
  \label{eq:LRlowerbounds}
   t_\mathrm{QFT} =
   \begin{cases}
      \W{1}, & \al = D
      \\ \Omega(\log r), & \alpha \in (D,2D]
      \\ \Omega(r^{(\alpha-2D)/(\alpha -D )}), & \alpha \in (2D,2D+1]
      \\ \Omega(r), & \alpha > 2D + 1.
    \end{cases}
  \end{align}
  For one-dimensional long-range systems, the Frobenius light cone gives the following tighter bounds \cite{Tran2020}:
  \begin{align}
  \label{eq:Froblowerbounds}
   t_\mathrm{QFT} =
   \begin{cases}
    \Omega(r), & \alpha > \frac{5}{2}
    \\ \Omega (r^{\alpha -3/2}/\log r), & \alpha \in (\frac{3}{2}, \frac{5}{2}].
   \end{cases}
  \end{align}
\end{theorem*}
We note that the lower bounds in the Theorem also apply to the fanout time, $t_\mathrm{FO}$, through the observation that fanout also performs operator spreading (using $X_1$ instead of $Z_1$).
We emphasize that these bounds pertain to the Hamiltonian model, where commuting terms can be implemented simultaneously and state transfer could in theory be done in $o(1)$ time for sufficiently small $\alpha$.

We observe that the QFT can implement quantum state transfer as well.
The goal of state transfer is to find a unitary $V$ such that
$V \bigl(\ket{\psi} \otimes \ket{0}^{\otimes n-1}\bigr)$\,$=$\,$\ket{0}^{\otimes n-1}$\,$\otimes $\,$\ket{\psi}$ \cite{Eldredge2017,Epstein2017}.
The unitary $V$\,$=$\,$H^{\otimes n} U_\mathrm{QFT}$ (where $H$ represents the single-qubit Hadamard gate) satisfies this definition of state transfer.

For the AQFT, the lower bound follows in a similar fashion.
The circuit that implements the QFT approximately with error $\varepsilon$ can be represented by a unitary $\tilde{U}_\mathrm{QFT}$ such that \cite{Cleve2000}
\begin{equation}
  \label{eq:AQFTdef}
  \|U_\mathrm{QFT}-\tilde {U}_\mathrm{QFT}\|\leq \varepsilon.
\end{equation}
Consider the operator $\tilde {U}_\mathrm{QFT}^\dag Z_1 \tilde {U}_\mathrm{QFT}$.
We argue that this operator is spread out as well.
From \cref{eq:AQFTdef}, it follows that
\begin{equation}
  \norm{\tilde {U}^\dag Z_1 \tilde {U} - U^\dag Z_1 U} \leq 2\varepsilon\norm{Z_1} = \O{\varepsilon},
\end{equation}
where $\| \cdot\|$ indicates the operator norm and we let $U$\,$\coloneqq$\,${U}_\mathrm{QFT}$ and $\tilde U$\,$\coloneqq$\,$\tilde {U}_\mathrm{QFT}$ for simplicity.
Since the normalized Frobenius norm is upper-bounded by the operator norm, we have
\begin{align}
\norm{\tilde {U}^\dag Z_1 \tilde {U} - U^\dag Z_1 U}_F = \O{\varepsilon}.
\end{align}
Moving to the vector space of operators and applying the projector $\mathcal{Q}_r$ onto operators with support beyond radius $r$ yields
\begin{align}
\norm{\mathcal{Q}_r |\tilde {U}^\dag Z_1 \tilde {U}) - \mathcal{Q}_r |U^\dag Z_1 U)} = \O{\varepsilon},
\end{align}
where $\|\cdot\|$ is the Euclidean norm and using $\|\mathcal{Q}_r\|$\,$=$\,$1$. By the triangle inequality, we have
\begin{align}
\label{eq:frobeniuslowerbound}
 \norm{\mathcal{Q}_r |\tilde {U}^\dag Z_1 \tilde {U})} & \geq \norm{\mathcal{Q}_r |U^\dag Z_1 U)} -\O{\varepsilon}
 \\ &= 1-\O{\varepsilon}.
\end{align}
\Cref{eq:frobeniuslowerbound} implies that the operator $Z_1$ after conjugating by the approximate QFT has large support on sites beyond distance $r$ as well, implying that the lower bounds in \cref{eq:LRlowerbounds,eq:Froblowerbounds} also hold for the approximate QFT.

\sect{Conclusions and Outlook}In
summary, we have developed a fast protocol for the unbounded fanout gate using power-law interactions.
For $\al$\,$\le$\,$D$, the protocol can perform the gate asymptotically faster than is possible with short-range interactions.
In particular, for experimentally realizable dipole-dipole interactions with $\al$\,$=$\,$3$, it allows the quantum Fourier transform and Shor's algorithm to be performed in logarithmic time in three dimensions.
As a corollary, we showed that classical simulation of strongly long-range systems with $\al$\,$\le$\,$D$ for time $t$\,$=$\,$\O{\log{n}}$ is at least as difficult as integer factorization, which is believed to be intractable in polynomial time.
Currently, the question of whether the fanout protocol is optimal remains open.
The best lower bound gives $\W{n^{\al/D-1}\log n} $ for $\al$\,$<$\,$D$  and $\W{1}$ for $\al$\,$=$\,$D$ \cite{Guo2019a}.
We conjecture that the broadcast time $\tghz$\,$=$\,$\O{\log n}$ is indeed the tightest that can be achieved for $\al$\,$=$\,$D$.

In addition, we gave a general $\W{n^{1/D}}$ lower bound on the time to implement fanout, as well as the exact or approximate QFT, for all systems constrained by a linear light cone.
In doing so, we used the state-of-the-art Frobenius bound from \cite{Tran2020}, which has been shown to be tighter than the Lieb-Robinson bound for certain long-range interacting systems in one dimension.
For higher dimensions, the conjectured critical value of $\alpha$ above which a linear light cone exists is $3D/2$\,$+$\,$1$.
If this generalization of the Frobenius bound were to hold, it would immediately tighten our lower bounds on the QFT and fanout.
Among other applications, this would imply the impossibility of implementing fanout in $o(n)$ time in cold-atom systems with van der Waals interactions ($\alpha$\,$=$\,$6$ in $D$\,$=$\,$3$ dimensions).
We also note the room for improvement in the catalogue of Lieb-Robinson bounds in  \cref{eq:LRlowerbounds}---especially for the power-law light cone between $\alpha$\,$\in$\,$(2D,2D+1]$.
The range of validity likely extends below $\al$\,$=$\,$2D$, and the exponent is suboptimal---at $\al$\,$=$\,$2D$\,$+$\,$1$, it is still a factor of $1/(D$\,$+$\,$1)$ from giving a linear light cone.

As a final remark, we have derived our lower bounds on $t_\textrm{QFT}$ under the assumption that the first and last qubits of the QFT are separated by a distance of $r$\,$=$\,$\Theta(n^{1/D})$.
However, other mappings of computational qubits to lattice qubits could potentially lead to faster implementations.
For example, consider the mapping onto a one-dimensional chain of qubits wherein the second half of the chain is interleaved in reverse order with the first half \cite{note_interleaving}.
Applying the QFT to a product state in this layout results in a state with two-qubit correlations that decay exponentially in the distance between the qubits.
In this case, our lower bound techniques cannot rule out the possibility of $t_\textrm{QFT}$\,$=$\,$o(n)$ for short-range interacting Hamiltonians.
This suggests that $t_\textrm{QFT}$ could depend strongly on qubit placement.
Given that the QFT is typically used as a subroutine for more complex algorithms, it may not always be possible to reassign qubits without incurring costs elsewhere in the circuit.
Still, it would be interesting to investigate whether careful qubit placement could yield a faster QFT.

\sect{Acknowledgements}We
would like to thank Minh Tran, Adam Ehrenberg, Chi-Fang Chen, Andrew Lucas, and Fred Chong for helpful discussions.
AYG would like to thank Dmitri Maslov and Yunseong Nam for informative discussions regarding the QFT.
AYG and DD are supported by the NSF Graduate Research Fellowship Program under Grant No.\ DGE-1840340.
We acknowledge support by the U.S. Department of Energy, Office of Science, Office of Advanced Scientific Computing Research, Quantum Testbed Pathfinder (award number DE-SC0019040) and Accelerated Research in Quantum Computing (award No.\ DE-SC0020312) programs, and the U.S. Army Research Office (MURI award number W911NF-16-1-0349).
AYG, AD, SKC, ZE, PB, and AVG additionally acknowledge support by AFOSR MURI, AFOSR,  DoE BES Materials and Chemical Sciences Research for Quantum Information Science program (award No.\ DE-SC0019449), NSF PFCQC program, ARL CDQI, and NSF PFC at JQI.
SKC also acknowledges the support from the Studying Abroad Scholarship by Ministry of Education in Taiwan.

\bibliography{Fanout,notes}

\begin{thebibliography}{49}
\bibitem [{big()}]{bigO}%
  \BibitemOpen
  \href@noop {} {}\bibinfo {note} {In this Letter, we use the following
  standard asymptotic notations: the ``big-$\mathcal{O}$'' notation $f(x) =
  \O{g(x)}$ indicates the existence of constants $c$ and $N$ such that $0\le
  f(x) \le cg(x)$ for all $x\ge N$; the ``big-$\Omega$'' notation $f(x) =
  \W{g(x)}$ indicates that $g(x) = \O{f(x)}$; the ``big-$\Theta$'' notation
  $f(x)=\Theta(g(x))$ when both $f(x)=\O {g(x)}$ and $f(x) = \W {g(x)}$;
  finally, the ``little-$o$'' notation $f(x) = o(g(x))$ indicates that
  $\lim_{x\to\infty} f(x)/g(x) = 0$. Additionally, we use $\text{poly}(n)$ to
  refer to computational resources that are polynomial in $n$.}\BibitemShut
  {Stop}%
\bibitem [{note_constanterror()}]{note_constanterror}%
  \BibitemOpen
  \bibinfo {note} {Showing hardness with any additive error lesser than some
  small constant is more stringent than showing hardness for
  $\varepsilon=\O{1/\text{poly}(n)}$.}\BibitemShut {Stop}%
\bibitem [{note_interleaving()}]{note_interleaving}%
  \BibitemOpen
  \bibinfo {note} {Specifically, for even $n$, this mapping is defined by
  $q_i\mapsto q_{2i-1}$ for $i \le n/2$ and $q_i\mapsto q_{2(n-i+1)}$ for $i >
  n/2$.}\BibitemShut {Stop}%
\end{thebibliography}
\end{document}